\documentclass{ocg}
\usepackage{amssymb}
\usepackage{amsmath}
\usepackage{latexsym}
\usepackage{textcomp}

\usepackage{comment}

\newcommand{\ket}[1]{|#1\rangle}

\newcommand{\cent}[0]{\mbox{\textcent}}
\newcommand{\dollar}[0]{\$}

\newcommand{\mymatrix}[2]{ \left( \begin{array}{#1} #2 \end{array} \right) }
\newcommand{\myvector}[1]{\mymatrix{c}{#1}}
\newcommand{\paran}[1]{\left( #1 \right)}

\newtheorem{fact}{Fact}

\begin{document}

\title{UNCOUNTABLE CLASSICAL AND QUANTUM COMPLEXITY CLASSES} 
\author[z]{Maksims Dimitrijevs} 
\author[e]{Abuzer Yakary\i lmaz} 

\address[z]{Faculty of Computing, University of Latvia, \\ Rai\c na bulv\= aris 19, Riga, LV-1586, Latvia\\
	\email{md09032@lu.lv}}
\address[e]{H\"urriyet Mah. 1755. Sok. 8/5, Yeni\c{s}ehir, Mersin, Turkey \\ 
\email{abuzer@boun.edu.tr}}

%

\maketitle

\begin{abstract}
Polynomial--time constant--space quantum Turing machines (QTMs) and logarithmic--space probabilistic Turing machines (PTMs) recognize uncountably many languages with bounded error (Say and Yakary\i lmaz 2014, arXiv:1411.7647). In this paper, we investigate more restricted cases for both models to recognize uncountably many languages with bounded error. We show that double logarithmic space is enough for PTMs on unary languages in sweeping reading mode or logarithmic space for one-way head. On unary languages, for quantum models, we obtain middle logarithmic space for counter machines. For binary languages, arbitrary small non-constant space is enough for PTMs even using only counter as memory. For counter machines, when restricted to polynomial time, we can obtain the same result for linear space. For constant--space QTMs, we follow the result for a restricted sweeping head, known as restarting realtime.
\end{abstract}

\section{Introduction}

As a well-known fact that two-wayness and alternation do not help to recognize a nonregular language for constant space Turing machines (finite state automata) \cite{Sze94}. When formally defined first time \cite{Rab63}, one--way probabilistic finite automata (PFAs) were also shown to recognize all and only regular languages with bounded error.  On the other hand, in his seminal paper, Freivalds \cite{Fre81} showed that two-way PFAs can recognize some nonregular languages with bounded error. But, it was shown that two--way PFAs require exponential expected-time to recognize nonregular languages \cite{DS90}. The quantum counterpart of two--way PFAs (two--way QFAs) were defined in \cite{KW97} and it was shown that they can recognize nonregular languages with bounded-error, even with one--way head move in linear time \cite{Yak12C}. Here the computational power comes from the input head being in a superposition called quantum head, which can be used to implement a counter in a very special way (see \cite{YFSA12A}). Later, two--way QFAs with classical head (2QCFAs) were defined \cite{AW02} and it was shown that they can recognize nonregular languages in polynomial expected time.

It is obvious that there are countably many regular languages since the description of any one--way deterministic finite automaton, which defines a single regular language, is finite. Similarly computable (regular or not) languages, the ones recognized by Turing machines, form a countable set. On the other hand, all languages form an uncountable set, which is an evidence of existing uncountably many nonregular languages.

On the other hand, a probabilistic or quantum model can be defined with uncomputable transiton values and so their cardinalities are uncountably many. Then, it is natural to ask whether they define an uncountable class. With unbounded-error (recognition with cutpoint), even unary 2-state QFAs and unary 3-state PFAs\footnote{As another ``probabilistic'' but unconventional model, ultrametric automata can also define uncountably many languages with 2 states \cite{MD16}.} define the classes formed by uncountably many languages \cite{SY14A,SY16A}. (Unary PFAs with 2-states define only a finite number of regular languages \cite{Paz71,SY16A}). So, the interesting case is investigating both models with bounded error: 
\newline \newline
\textit{What are the minimal bounded-error probabilistic and quantum classes that contain uncountably many languages?}

Polynomial--time constant--space quantum Turing machines (QTMs) and logarithmic--space probabilistic Turing machines (PTMs) are known to recognize uncountably many languages \cite{ADH97,SayY14C}. In this paper, we investigate more restricted cases for QTMs and PTMs, i.e. using less space, restricted memory types, and restricted input head moves.

We show that double logarithmic space is enough for PTMs on unary languages in sweeping reading mode or logarithmic space for one--way head. On unary languages, for quantum models, we obtain middle logarithmic space for counter machines. For binary languages, arbitrary small non-constant space is enough for PTMs even using only counter as memory. For counter machines, when restricted to polynomial time, we can obtain the same result for linear space. For constant--space QTMs, we follow the result for a restricted sweeping head, known as restarting realtime.

In the next section, we give the required background with a short introduction to quantum operators and then we present our results in Section \ref{sec:main-results} under four subsections. We first present the results for PTMs that are pedagogically easy to follow (Section \ref{sec:PTM}). Then, we restrict the model to use counter as a memory (Section \ref{sec:PCA}). After this, we also check some restrictions on the head movement (Section \ref{sec:1PTM}). Lastly, we present our quantum results (Section \ref{sec:QCFA}). We close the paper by listing all results with possible future directions in Section \ref{sec:conc}.

\section{Background}

\newcommand{\sigmastar}{\Sigma^{*}}
\newcommand{\tildesigma}{\tilde{\Sigma}}
\newcommand{\tildegamma}{\tilde{\Gamma}}
\newcommand{\tildew}{\tilde{w}}

We assume the reader is familiar with the basics of complexity theory and automata theory.  We refer the reader to \cite{NC00} for a complete reference on quantum computation, to \cite{SayY14} for a pedagogical introduction to QFAs, and to \cite{AY15} for a comprehensive chapter on QFAs.

Throughout the paper, $ \# $ denotes the blank symbol, $ \varepsilon $ denotes the empty string, $ \Sigma $ not containing $\cent$ (the left end-marker) and $\dollar$ (the right end-marker) denotes the input alphabet, $ \tildesigma $ is the set $ \Sigma \cup \{ \cent,\dollar \} $, $ \Gamma $ not containing $ \# $  denotes the work tape alphabet, $ \tildegamma $ is the set $ \Gamma \cup \{ \# \} $, and $ \sigmastar $ is set of all strings obtained from the symbols in $\Sigma$ including the empty string. We order the elements of $ \sigmastar $ lexicographically and then represent the  $i$-th element by $ \sigmastar(i) $ where the first value $ \sigmastar(1) $ is the empty string. 
 We fix $ n $ as the length of any given input.

Each model has a read--only one--way infinite input tape with a single head, on which the given input $ w $ is placed as $ \tilde{w} $. All the remaining tape cells are filled with blank symbols. At the beginning of the computation, the input head is placed on the left end-marker. Each machine is designed to guarantee that the input head never visits outside $ \tilde{w} $. A work tape is a two--way infinite tape with a single head, where each tape cell is indexed with an integer. At the beginning of the computation, all cells of a work tape are filled with symbols $\#$ and the head is placed on the cell indexed by zero.

A deterministic Turing machine (DTM) $ D $ having an input tape and a work tape is a 7-tuple
\[
	D = ( S,\Sigma,\Gamma,\delta,s_1,s_a,s_r ),
\]
where $ S $ is the set of finite internal states, $ s_1 \in S $ is the initial state, $ s_a \in S $ and $ s_r \in S $ ($s_a \neq s_r$) are the accepting and rejecting states, respectively, and $ \delta $ is the transition function
\newcommand{\directions}{ \{ \leftarrow,\downarrow,\rightarrow \} }
\[
	\delta: S \times \tildesigma \times \tildegamma \rightarrow S \times \tildegamma \times \directions \times \directions
\] 
that governs the behaviours of $D $ as follows: When $ D $ is in state $ s \in S $, reads symbol $ \sigma \in \tildesigma $ on the input tape, and reads symbol $ \gamma \in \tildegamma $ on the work tape, it follows the transition
\begin{equation}
	\label{eq:delta}
	\delta(s,\sigma,\gamma) = (s',\gamma',d_i,d_w),
\end{equation}
and then the state becomes $ s' \in S $, $ \gamma' $ is written on the cell under the work head, and then the positions of input and work heads are updated with respect to $ d_i \in \directions $ and $ d_w \in \directions $, respectively, where ``$ \leftarrow $'' (``$\downarrow$'' and ``$\rightarrow$'') means the head is moved one cell to the left (the head does not move and the head is moved one cell to the right). The computation starts in state $ s_1 $, and the computation is terminated and the given input is accepted (rejected) if $ D $ enters $ s_a $ ($s_r$). The set of strings accepted by $ D $ form a language, say $ L \subseteq \Sigma^* $, and it is said that $ L $ is recognized by $ D $. 

The space used by $ D $ on a given input is the number of all cells visited on the worktape during the computation. 

If the input head is not allowed to move to the left, then it is called ``one--way'' and then the model is denoted as 1DTM. If we remove the work tape (and all related components in the formal definition and in the transition function) of a DTM/1DTM, we obtain a two--way/one--way deterministic finite automaton (2DFA/1DFA). 

A counter is a special type of memory containing only the integers. Its value is set to zero at the beginning. During the computation, its status (whether its value is zero or not) can be read like reading blank symbol or not on the work tape, and then its value is incremented or decremented by 1 or not changed like the position update of work head. Thus a deterministic counter automaton (2DCA) is a 2DFA with a counter. The space used (on the counter) by a 2DCA is maximum value of the counter during the computation. Remark that the value of the counter can be stored on a binary work tape with logarithmic amount of the space. Moreover, counter can also be seen as a unary work tape with certain specifications.

A probabilistic Turing machine (PTM) is a generalization of a DTM such that it can make random choices according to some probability distributions and so a PTM can do more than one transition in each step. Each choice can be realized only with some probabilities. The number of choices and their realization probabilities are determined by the current state and the symbols read on the tapes, and the summation of the probabilities must be 1 to have a well--formed probabilistic systems. Thus, a PTM can follow different paths during the computation and so the input is accepted with some probabilities. Remark that all probabilistic models in this paper halt either absolutely or with probability 1. In the latter case, we mention about expected running time. Since the space usage of a PTM can be different on the different paths, we take the maximum value.

The language $L$ is said to be recognized by PTM with error bound $\epsilon$ ($0 \leq \epsilon < 1/2$) if every member of $ L $  is accepted with probability at least $1-\epsilon$ and every non-member of $L$ ($w \notin L$) is accepted with probability not exceeding $\epsilon$.

One-way PTM (1PTM) is defined similar to 1DTM. If we remove the work tape of a PTM/1PTM, we obtain two--way/one--way probabilistic finite automaton (2PFA/1PFA). A probabilistic counter automaton (2PCA) is a 2PFA with a counter. 

An $m$-state ($Q=\{q_1,\ldots,q_m\}$) quantum system forms an $m$-dimensional Hilbert space ($\mathcal{H}^m$), complex vector space with inner product, spanned by the set $ \{ \ket{q_1},\ldots,\ket{q_m} \} $, where $ \ket{q_j} $ is a column vector with zero entries except the $ j $-th entry that is 1. A quantum state of the system is a norm-1 vector in $ \mathcal{H}^m $:
\[
	\ket{v} = \alpha_1 \ket{q_1} + \cdots + \alpha_m \ket{q_m},~~ \sum_{j=1}^m |\alpha_j|^2 = 1,
\]
where $ \alpha_j $ is a complex number and represents the amplitude of the system being in $ \ket{q_j} $, and the probability of system being in $ \ket{q_j} $ is given by  $ |\alpha_j|^2 $.

The quantum system evolves by unitary operators, also known as norm preserving operators, represented by unitary matrices. Let $ U $ be a unitary operator (matrix). Then its $ (l,j) $-th entry represents the transition amplitude from $ \ket{q_j} $ to $ \ket{q_l} $, where $ 1 \leq j,l \leq n $. After applying $ U $, the new state is
\[
	v' = U v = \alpha'_1 \ket{q_1} + \cdots + \alpha'_m \ket{q_m},~~ \sum_{j=1}^m |\alpha'_j|^2 = 1.
\]

In order to retrieve information from the system, measurement operators are applied. We use a simple one called projective measurement,  say $P$. Formally $P$ is composed by $ k \geq 1 $ elements $ \{P_1,\ldots,P_k\} $. Each $ P_i $ is a zero-one diagonal (nonzero) matrix and $ P_1 + \cdots + P_k = I $. So $P$ is designed to decompose $ \mathcal{H}^m $ into $ k $ orthogonal subspaces and $ P_j $ projects any vector to its subspace, where $ 1 \leq j \leq k $. After applying $P$ to the system when in $ \ket{v} $, the system collapses to one of the subspaces and so the new quantum state lies only in this subspace. The vector
\[
	\ket{\widetilde{v_j}} = P_j \ket{v}
\] 
is the projection of the quantum state to the $j$-th space and so the probability of observing the system in this subspace is given by $ p_j  = || \ket{ \widetilde{v_j} } ||^2 $. If this happens ($p_j >0$), the new quantum state is
\[
	\ket{v_j} = \frac{\ket{\widetilde{v_j}}}{\sqrt{p_j}}.
\]
The vector $ \ket{\widetilde{v_j}} $ is called unnormalized state vector and tracing quantum systems by such vectors can make the calculations simpler. (Keeping the conditional probabilities with normalized states may make the calculations harder to follow.)

Now we give the definition of two--way quantum finite automaton with classical head, known as two--way finite automaton with quantum and classical states (2QCFA) \cite{AW02}, which can use unitary operators and projective measurements on the quantum part. Formally, a 2QCFA $ M $ is a 8-tuple
\[
	M = ( S,Q,\Sigma,\delta,s_1,q_1,s_a,s_r ),
\]
where, different from a classical model, $Q$ is the set of quantum states, $q_1$ is the initial state, and the transition function $ \delta $ is composed by $ \delta_q $ governing quantum part and $ \delta_c $ governing the classical part. The computation is governed classically. At the beginning of the computation, the classical part is initialized and the state of quantum part is set to $ \ket{q_1} $. In each step, the current classical state and scanned symbol determines a quantum operator, either a unitary operator or a projective measurement, that is applied to the quantum register. After getting the new quantum state, the classical part is updated. If the quantum operator is unitary, then classical part is updated like a 2DFA. If the quantum operator is a measurement, then the outcome is processed classically, that is, the next state and head movement is determined by the current classical state, the measurement outcome, and the scanned symbol. When entering $ s_a $ ($s_r$), the computation halts and the input is accepted (rejected). 

A (strict) realtime version of 2QCFA (rtQCFA) \cite{ZQLG12} moves its head one square to the right in each step, halts the computation after the reading the right end-marker, and applies one unitary operator and then measurement operator for the quantum part in each step. 

A 2QCFA with counter (2QCCA) is a 2QCFA augmented with a classical counter, where the classical part can access a counter.

A two--way model is called sweeping if the direction of the head can be changed only on the end-markers. So, the input is read from left to the right, then right to left, and then left to right, and so on. A very restricted version of sweeping models are realtime restarting models (see \cite{YS10B,YS11D} for the details of restarting concept): the models have an additional state $ s_i $ such that immediately after entering $s_i$ the overall computation is terminated and all computations start from the initial configuration. In this paper, we focus on restarting rtQCFAs.

We denote the set of integers $ \mathbb{Z} $ and the set of positive integers $ \mathbb{Z}^+ $. The set $ \cal I $ is the set of all subsets of $ \mathbb{Z^+} $:
\[
	\mathcal{I} = \{ I \mid I \subseteq \mathbb{Z^+} \}.
\]
Remark that the cardinality of $ \mathbb{Z} $ or $ \mathbb{Z^+} $ is $ \aleph_0 $ (countably many) and the cardinality of $ \mathcal{I} $ is $ \aleph_1 $ (uncountably many) like the set of real numbers ($ \mathbb{R} $). The membership of each positive integer in any $ I \in \mathcal{I} $ can be represented as a binary probability value:
\[
	p_I = 0.x_1 0 1 x_2 0 1 x_3 0 1 \cdots x_i 0 1 \cdots,~~~~ x_i = 1 \leftrightarrow i \in I.
\]
Quantumly, we use a different technique, originally given in \cite{ADH97}. The membership of each positive integer in any $ I \in \mathcal{I} $ can be represented as a single rotation on $ \mathbb{R}^2 $ with the angle:
\[
	\theta_I = 2 \pi \sum_{i=1}^\infty \left( \frac{x_i}{8^{i+1}} \right) ,~~~~ \begin{array}{lrl} x_i = & 1, & \mbox{if } i \in I \\ x_i =  & -1, & \mbox{if } i \notin I	
\end{array}	 .
\]

\section{Main Results}
\label{sec:main-results}

We start with PTMs. Then we focus on 2PCAs and PTMs with restricted head movements. Lastly, we present our quantum results.

\subsection{Probabilistic Turing machines}
\label{sec:PTM}

It is known that polynomial--time PTMs can use uncomputable transition probabilities to recognize uncountably many languages with bounded error \cite{SayY14C}. The $k$-th bit in the decimal expansion of the probability that a given biased coin will land heads can be estimated by a procedure that involves tossing that coin for a number of times that is exponential in $k$. Given any unary language $L$ on the alphabet $ \{a\}$, and a coin which lands heads with probability $0.x$, where $x$ is an infinite sequence of digits whose $k$-th member encodes whether the $k$-th unary string is in $L$, the language $\{a^{4^k}|a^k \in L\}$ is recognized by PTM with bounded error. The machine in this construction uses logarithmic space.

In this section, we improve this result and show that bounded-error probabilistic models can recognize uncountably many languages with less resources. We start with a technical lemma.

\begin{lemma}
	\label{lem:64k}
	Let $ x = x_1 x_2 x_3 \cdots $ be an infinite binary sequence. If a biased coin lands on head with probability  $p=0.x_101x_201x_301...$, then the value $ x_k $ can be determined with probability $ \frac{3}{4} $ after $64^k$ coin tosses.
\end{lemma}
\begin{proof}
	Let $X$ be the random variable denoting the number of heads after $64^k$ coin flips. The expected value of $X$ $E[X]=p*64^k$. The value of $x_k$ is equal to $(3*k-2)$-th bit in $E[X]$. 
	
	If $|X-E[X]| \leq 8^k$ we still have the correct $x_k$, because $E[X]=x_101x_201x_301...x_k01...$ and is followed by $3k$ bits after $x_k01$, and: if we add up to $8^k$ to this number, we get at most $x_101x_201x_301...x_k10...$, followed by $3k$ bits after $x_k10$; if we subtract up to $8^k$ from this number, we get at least $x_101x_201x_301...x_k00...$, followed by $3k$ bits after $x_k00$. In both mentioned cases the bit $x_k$ remains unchanged. This means that the probability of error does not exceed $Pr[|X-E[X]| \geq 8^k]$. By Chebyshev’s inequality we can get that 
	\[
		Pr[|X-E[X]| \geq 8^k] \leq \frac{p*(p-1)*64^k}{(8^k)^2}=\frac{p*(p-1)*64^k}{64^k}=p*(p-1).
	\]
	Therefore, the probability of an error is at most $p*(p-1)$. Function $p*(p-1)$ is parabolic function and its global maximum is $\frac{1}{4}$. That is, $ p*(p-1) \leq \frac{1}{4} $ for any chosen probability $p$. Therefore, the procedure gives the correct answer with the probability at least $\frac{3}{4}$.
\end{proof}

Now, we show that $ O(\log \log n) $ space is enough to recognize uncountably many languages.

\begin{theorem}
	\label{thm:poly-unary-PTM}
	Polynomial--time bounded--error unary PTMs can recognize uncountably many languages in $ O(\log \log n) $ space.
\end{theorem}
\begin{proof}
	For our purpose, we define languages based on the following unary language given by Alt and Mehlhorn in 1975 \cite{AM75}:
	\[
		\mathtt{AM75}=\{a^n \mid n>0 \mbox{ and } F(n) \mbox{ is a power of } 2\},
	\]
	 where $F(n)=min\{i \mid i \mbox{ does not divide } n \} \in \{2,3,4,\ldots\}$. It is known that $O(\log \log n)$-space DTMs can recognize $ \mathtt{AM75} $.  It is clear that the following language
	 \[
	 	\mathtt{AM75'}=\{a^n \mid n>0 \mbox{ and } F(n) \mbox{ is a power of } 64\}
	 \] 
	 can also be recognized by $O(\log \log n)$-space DTMs. For the sake of completeness, we provide the details of the algorithm (see also \cite{Sze94}).
	
	Assume that the input is $w=a^n$ for some $ n > 0 $. In order to check if a number $k$ written in binary on the work tape divides $n$, we can use $ O(\log k) $ space for binary values of $ k $ that form a counter, and then we can check whether $ n \mod k $ is equal to zero or not. In order to compute $F(n)$, we can check each $ k = 1,2,3,\ldots $ in order to determine the first $ k $ such that $ n \mod k \neq 0 $. It is known that (see Lemma 4.1.2(d) in \cite{Sze94}), $F(n)<c*\log n$ for some constant $c$. Therefore, we use $O(\log \log n)$ space to find $ F(n) $. Remark that when the number $F(n)$ is found, it is written on the work tape, and it is easy to check whether this number is a power of 64, i.e., it must start with 1 and should be followed by only zeros and the number of zeros must be a multiple of 6 ($64=2^6$). 
	
	For any $ I \in \mathcal{I} $, we can define a corresponding language:
	\[
		\mathtt{AM75'(I)} = \{ a^n \mid a^n \in \mathtt{AM75'} \mbox{ and } \sigmastar(\log_{64}F(n)) \in I \}.
	\]  
	
	For any input $ a^n $, we can deterministically check whether $ a^n \in \mathtt{AM75'} $ by using the above algorithm. If not, the input is rejected. Otherwise, we continue with a probabilistic procedure. Remark that the work tape can still contain the binary value of $ F(n) $ that is $ 64^m $ for some positive integers $ m $ in the beginning of the probabilistic procedure.
	
	We use a biased coin landing on head with probability $ p_I $ encoding the memberships of positive integers in $I$ as described before.
	
	By definition we know that $ a^n \in \mathtt{AM75'(I)} $ if and only if $ m \in I $. So, if we compute the value of $ x_m $ correctly, we are done. Since the work tape contains the value of $ 64^m $, we can toss this biased coin $64^m$ times and count the number of heads. Due to Lemma \ref{lem:64k}, we know that we can correctly compute $ x_m $ with probability at least $ \frac{3}{4} $. Here the number of heads is kept in binary and we check the $(3m-2)$-th bit of the result after finishing the all coin tosses. By executing the probabilistic procedure a few more times, the success probability can be increased. Remark that the space used on the work tape does not exceed $ O(\log \log n) $ and so the running time is polynomial in $ n $.
	
	The cardinality of the set of all subsets of positive integers is uncountably many and so the cardinality of the following set
	\[
		\left\lbrace \mathtt{AM75'(I)} \mid I \subseteq \mathbb{Z}^+ \right\rbrace
	\]
	is also uncountably many, each element of which is recognized by a polynomial-time bounded-error unary PTM using $ O( \log \log n) $ space.
\end{proof}

With polynomial expected time, we cannot do better since it was proven that polynomial-time PTMs using $o(\log \log n)$ space can recognize only regular languages even with unrestricted transition probabilities \cite{DS90}. 

On the other hand, a well-known fact is that with super-polynomial expected time PTMs can recognize nonregular binary languages even with constant-space \cite{Fre81}. Remark that constant-space unary PTMs can recognize only regular languages \cite{Kan91B} and regarding $ o(\log \log n) $ space, we only know that one-way unary PTMs cannot recognize any nonregular language \cite{KF90}. Currently we leave open whether constant--space PTMs can recognize uncountably many languages and we do not know whether PTMs can recognize a unary nonregular language with $ o(\log \log n) $ space. But we show that PTMs can recognize uncountably many (binary) languages with arbitrary small non-constant space. For this purpose, we use a fact given by Freivalds in \cite{Fre81}. Note that here we use a slightly modified version of the original language given in \cite{Fre81} in order to keep the input alphabet binary.

For any binary language $ L \subseteq \{0,1\}^* $, we define another language $ \mathtt{LOG(L)} $ as follows:
\[
	\mathtt{LOG(L)} = \{ 0 (1 w_1) 0^{2^1} (1 w_2) 0^{2^2} (1 w_3) 0^{2^3} \cdots 0^{2^{m-1}} (1 w_m) 0^{2^m} \mid w = w_1 w_2 \cdots w_m \in L \}.
\]

\begin{fact}
	\label{fact:Fre81}
	\cite{Fre81}
	If a binary language $ L $ is recognized by a bounded-error PTM in space $ s(n) $, then the binary language $ \mathtt{LOG(L)} $ is recognized by a bounded-error PTM in space $ \log(s(n)) $.
\end{fact}

\begin{theorem}
	\label{thm:log-AM75}
	For any $ I \in \mathcal{I} $, the language $ \mathtt{LOG(AM75'(I))} $ can be recognized by a bounded-error PTM in space $ O(\log \log \log (n)) $. 
\end{theorem}
\begin{proof}
	It follows from Theorem \ref{thm:poly-unary-PTM} and Fact \ref{fact:Fre81}. 
\end{proof}

Similarly we can follow that the language $ \mathtt{LOG^k(AM75'(I))} $ for $ k>1 $ can be recognized by a bounded-error PTM in space $ O(\log^{k+2}(n)) $. 

\begin{corollary}
	\label{cor:PTM-arbitrary-small}
	The cardinality of languages recognized by bounded-error PTMs with arbitrary small non-constant space bound is uncountably many.
\end{corollary}

\subsection{Probabilistic counter machines}
\label{sec:PCA}

In this section, we present some results for 2PCAs. Remark that any $ s(n) $-space counter can be simulated by $\log(s(n))$-space work tape.

It is easy for a 2PCAs to check whether any specific part of the input has length of $ 64^k $ for some $ k>0 $, and so, they can easily toss a biased coin for $ 64^k $ times and then count the number of heads on the counter. However, it is not trivial to read some certain digits of the result on the counter and so we use a clever trick here.

\begin{theorem}
	\label{thm:linear-2PCA}
	Bounded-error linear-time (linear-space) 2PCAs can recognize uncountably many languages.
\end{theorem}
\begin{proof}
	We start with the definition of a new language:
	\[
		\mathtt{DIMA} = \{ 0^{2^0}10^{2^1}10^{2^2}1  \cdots 1 0^{2^{3k+1}}110^{2^{3k+2}}110^{2^{3k+3}}1 \cdots 10^{2^{6k}} \mid k > 0 \}.
	\]
	Remark that each member is composed by $ (6k+1) $ zero-blocks separated by single 1s except two special separators 11 that are used as the marker to indicate the $ (3k+3) $-th block, the length of which is $ 2^{3k+2} $.
	
	The language $ \mathtt{DIMA} $ can be recognized by a 2DCA, say $D$. First it checks that the input starts with a single 0 and then ends with some 0s, all separators are 1s except two of them, which are 11 and consecutive, and the number of zero-blocks is $ 6k+1 $ for some $k>0$. For all these checks, $ D $ can use only its internal states. And then by using its counter it can check the length of each zero-block (except the first one) is double of the length of previous block. Similarly, it can check the equality of the number of zero-blocks before the first ``11'' and the number of zero-blocks after the first ``11'' plus 3, i.e. $ 3k+2 $ versus $ (3k-1) + 3 $. If one of these checks fails, then the input is rejected immediately. Otherwise, it is accepted. Remark that $ D $ can finish its computation in linear time and the counter value never exceeds the input length.
 	 
	For any $ I \in \mathcal{I} $, we define a new corresponding language:
	\[
		\mathtt{DIMA(I)} = \{ w \in \{0,1\}^*10^{m} \mid m>0, w \in \mathtt{DIMA},  \mbox{ and } \sigmastar(\log_{64}m) \in I \}.
	\]
	For any such $ I $, we can construct a 2PCA recognizing $ \mathtt{DIMA(I)} $, say $P_I$, as desired. The machine $ P_I $ checks whether any given input, say $w$, is in $ \mathtt{DIMA} $ deterministically by using $D$. If the input is not rejected by $D$, we continue with a probabilistic procedure. Since the last zero-block has the length of $ m = 64^{k} $, by reading this block $ P_I $ can toss $ 64^{k} $ biased coins that land on head with probability $ p_I $. The number of heads are counted on the counter. Similar to the proof of Theorem \ref{thm:poly-unary-PTM}, the only remaining task is to determine the $ (3k-2) $-th bit of the binary value of the counter, which is $x_k$. The bit $x_k$ in $E[X]=p_I*64^k$ is followed by $3k+2$ bits. 
	
	The number of heads on the counter, say $C$, can be written as a binary number as follows:
	\[
		C = \sum_{i=0}^{6k} a_i 2^{i} = a_{6k}2^{6k} + \cdots + a_{3k+2} 2^{3k+2} +  a_{3k+1} 2^{3k+1} + \cdots + a_1 2 + a_0,
	\]
	where each $ a_i \in \{0,1\} $. Remark that $ x_k = a_{3k+2} $, i.e. $ 3k+2 = 6k - (3k-1) + 1 $. We can rewrite $ C $ as
	\[
		C = B_1 2^{3k+3} +  a_{3k+2} 2^{3k+2} + B_0 = B_1 2^{3k+3} + C',
	\]
	where $ B_0 $ and $ B_1 $ are integers, $ B_0 < 2^{3k+2} $, and $ C' = a_{3k+2} 2^{3k+2} + B_0 $.
	
	After tossing-coin part, $ P_I $ moves its head to the second symbol of the first ``11'' and then the automaton enters a loop. In each iteration, the head moves to next separator on the right by reading $ 2^{3k+2} $ $0$s and then comes back by reading the same amount of 0s. In each iteration, $ P_I $ tries to subtract $ 2^{3k+2} $ twice ($ 2^{3k+3} $). 
	
	If $ C'=0 $, then $ P_I $ hits to the zero value on the counter when the head is at the starting position of the loop. This means $ a_{3k+2} = x_k = 0 $ and so the input is rejected by the automaton $ P_I $. If $ C' \neq 0 $, then $ P_I $ hits to the zero value on the counter (at some $j$-th iteration, $ j=0,1,\ldots) $ when the head is not at the starting position of the loop. It is clear that the value of counter is $ C' $ before starting the $j$-th iteration. Now, we have two cases, $ x_k = 1 $ or $ x_k = 0 $. If $ x_k = 1 $, $ P_I $ hits to the zero value on the counter only after finishing to read the first  $ 2^{3k+2} $ $0$s. In this case, the input is accepted. Otherwise, $ P_I $ hits to the zero value on the counter before finishing to read the first  $ 2^{3k+2} $ $0$s. Then, the input is rejected.
	
	It is clear that the value of counter never exceeds length of the input. Moreover, both deterministic and probabilistic parts finish in linear time.
\end{proof}

By relaxing the linear-time, we follow similar results for arbitrary small non-constant space on the counter like PTMs.

\begin{theorem}
	\label{thm:log-DIMA}
	For any $ I \subseteq \mathcal{I} $, the language $ \mathtt{LOG(DIMA(I))} $ can be recognized by a bounded-error 2PCA that uses $ O(\log (n)) $ space on the counter. 
\end{theorem}
\begin{proof}
	Let $ R'_I $ be our desired 2PCA. The definition of $ \mathtt{LOG(DIMA(I))} $ is 
	 \[
		\{ 0 (1 w_1) 0^{2^1} (1 w_2) 0^{2^2} (1 w_3) 0^{2^3} \cdots 0^{2^{m-1}} (1 w_m) 0^{2^m} \mid w = w_1 w_2 \cdots w_m \in \mathtt{DIMA(I)} \}.
	\]
	The automaton can deterministically check the input of the form $ (0^+1\{0,1\}^*)0^+ $. If not, the input is rejected. If so, we can assume that the input is of the form
	\[
		0^+ (1 w_1) 0^{+} (1 w_2) 0^{+} (1 w_3) 0^{+} \cdots 0^{+} (1 w_m) 0^{+}
	\] and the computation continues.
	
	If we are sure that each zero block (except the first one) has double length of the previous zero block, $ R'_I $ executes $ R_I $ on $ w = w_1 w_2 \cdots w_m $ by giving the same answer as $ R_I $ and so we are done. In such a case, $ R_I $ uses linear space on the counter in $ m $, which is logarithm of the input length. 
	
	It is clear that if we use the counter to compare the length of zero blocks in regular way, then the value of counter cannot be sub-linear. On the other hand, as shown by Freivalds \cite{Fre81}, 2PFAs can make such a sequence (unary) equality checks with high probability (see Lemma 2 in \cite{Fre81}). (The only drawback is that 2PFAs require exponential expected time for these checks \cite{DS90}). 
	
	Thus, after the first deterministic check, $ R'_I $ determines the well form of zero blocks with high probability without using its counter. If the zero blocks are well formed, it calls $ R_I $ on $ w $. Otherwise, the input is rejected.
\end{proof}

Similarly we can follow that the language $ \mathtt{LOG^k(DIMA(I))} $ for $ k>1 $ can be recognized by a bounded-error 2PCA that uses $ O(\log^{k}(n)) $ space on the counter. 

\begin{corollary}
	\label{cor:2PCA-arbitrary-small}
	The cardinality of languages recognized by bounded-error 2PCAs with arbitrary small non-constant space bound is uncountably many.
\end{corollary}

\subsection{One-way and sweeping probabilistic machines}
\label{sec:1PTM}

Here we present some results by assuming further restrictions.

\begin{theorem}
	\label{thm:1PTM-log}
	Linearithmic--time  bounded--error 1-way unary PTMs can recognize uncountably many languages in $ O( \log n) $ space.
\end{theorem}
\begin{proof}
	We start with the definition of language $ \mathtt{UPOWER64} $:
	\[
		\mathtt{UPOWER64} = \{ 0^{2^{6k}}| k > 0 \}.
	\] 
	This language is recognized by 1DTMs in $ O( \log n) $ space, where $n$ is the length of the input. A binary counter on the work tape is used to count the number of zeros in the input. This can be done in a straightforward way. For each input symbol, the value of the counter is increased by 1. Remark that any update on the counter can be done in $ O(\log n) $ steps. Once the whole input is read, the counter is checked whether it is a power of $ 64 $, i.e. it must start with 1 and should be followed by only zeros and the number of zeros must be a multiple of 6 ($64=2^6$). The overall running time is $ O(n \log n) $.
	
	Like  in Theorem \ref{thm:poly-unary-PTM}, for any $ I \in \mathcal{I} $, we can define a corresponding language:
\[
	\mathtt{UPOWER64(I)} = \{ 0^n \mid 0^n \in \mathtt{UPOWER64} \mbox{ and } \sigmastar(\log_{64}n) \in I \}.
\] 
We again use a biased coin landing on head with probability $ p_I $.
When we read the input and count the number of zeros, we can in parallel toss the biased coin and count the number of heads in a second counter on the working tape. After reading the whole input, for the inputs in $ \tt UPOWER64 $, the decision is given by checking the $(3m-2)$-th bit of the second counter. This additional probabilistic procedure does not change the runtime and space asymptotically.
\end{proof}

The algorithm given in the proof of Theorem \ref{thm:poly-unary-PTM} for the language
	\[
		\mathtt{AM75'(I)} = \{ a^n \mid a^n \in \mathtt{AM75'} \mbox{ and } \sigmastar(\log_{64}F(n)) \in I \}
	\]
	does not need to change the direction of head on the $a$s. So, we can call that PTM sweeping. 
\begin{corollary}
	\label{cor:sweep-PTM}
	Polynomial--time bounded--error sweeping unary PTMs can recognize uncountably many languages in $ O(\log \log n) $ space.
\end{corollary}

\begin{theorem}
	Bounded--error linear--space sweeping PCAs can recognize uncountably many languages in subquadratic time.
\end{theorem}
\begin{proof} 
	We modify the algorithms given in the proof of Theorem \ref{thm:linear-2PCA}. Remark that the algorithms given there run in linear time. Here the algorithms run in super-linear time. First we show how to deterministically recognize the language $ \mathtt{DIMA} $ in sweeping reading mode, i.e.,
	\[
	\mathtt{DIMA} = \{ 0^{2^0}10^{2^1}10^{2^2}1  \cdots 1 0^{2^{3k+1}}110^{2^{3k+2}}110^{2^{3k+3}}1 \cdots 10^{2^{6k}} \mid k > 0 \}.
	\]
	
	With one pass (reading the input from the left end-marker to the right end-marker), the input is checked without using counter whether having the following form
	\[
		01(0^+1)^+ 11 0^+ 11 (0^+1)^+0^+
	\] 
	and the number of 0-blocks are $ 6k+1 $ for some $ k> 0 $. Moreover, for a member, the number of 0-blocks before the first ``11'' is $ 3k+2 $ and  the number of 0-blocks after the second ``11'' is $ 3k-2 $. Therefore, by using the counter, we can check that the number of 0-blocks before the first ``11'' is 4 more than the number of 0-blocks after the second ``11''. If any of these checks fails, then the input is immediately rejected.
	
	In the second pass (reading input from the right end-marker to the left end-marker), it is checked that, for each $0 < i \leq 3k$, $(2i+1)$-th 0-block has twice more zeros than $(2i)$-th 0-block.
	
	In the third pass (reading input from the left end-marker to the right end-marker), it is checked that, for each $0 < i \leq 3k$, $(2i-1)$-th 0-block has twice less zeros than $(2i)$-th 0-block.
	
	Thus, in three passes, $ \mathtt{DIMA} $ can be recognized by a sweeping PCA.
	
	Then, as in the proof of Theorem \ref{thm:linear-2PCA}, for any $ I \in \mathcal{I} $, now we focus on the language:
	\[
	\mathtt{DIMA(I)} = \{ w \in \{0,1\}^*10^{m} \mid m>0, w \in \mathtt{DIMA},  \mbox{ and } \sigmastar(\log_{64}m) \in I \}.
	\]	
	If the given input is in $ \mathtt{DIMA} $, then we continue with the probabilistic procedure. (Otherwise, the input is rejected.) We perform the same walk as in the proof of Theorem \ref{thm:linear-2PCA}, but, due to sweeping reading mode, each walk can be done from one end-marker to the other end-marker. But the presence of symbols ``11'' allows us to follow the same procedure only with slowdown. The running time is $ O(2^{3k}) O(2^{6k}) = O(2^{9k}) $ and it is super-linear and subquadratic in the length of input. To be more precise, the running time is $ O(n\sqrt{n}) $ if the $n$ is the length of the input.
\end{proof}

\subsection{Quantum models}
\label{sec:QCFA}

For any $ I \in \mathcal{I} $, we can compute the membership of the positive integer $ j $ in $I$ as described below \cite{ADH97,SayY14C}. We call it Procedure \textit{ADH}.

The qubit spanned by $ \{ \ket{q_1},\ket{q_2} \} $ is set to $ \ket{q_1} $. Then, it is rotated with angle $ \theta_I $  $ 8^j $ times, which leaves the quantum state having angle
\[
	8^j \cdot 2 \pi \sum_{i=1}^\infty \left( \frac{x_i}{8^{i+1}} \right) = \pi \left( \frac{x_j}{4} \right) + \sum_{i=j+1}^\infty \left( \frac{x_i}{8^{i+1}} \right)
\]
from the initial position. After an additional rotation by $ \frac{\pi}{4} $, the final angle from $ \ket{q_1} $ is $\frac{\pi}{2} + \delta$ if $ x_i=1 $ ($ i \in I $) and it is $ \delta $ if $ x_i=-1 $ ($i \notin I$), where $ \delta $ is sufficiently small such that the probability of the qubit being in $ \ket{q_2} $ ($ \ket{q_1} $) is bigger than $ 0.98 $ if $ i \in I $ ($i \notin I$).

\newcommand{\powereq}{\mathtt{POWER\mbox{-}EQ}}
\newcommand{\powereqI}{\mathtt{POWER\mbox{-}EQ(I)}}

Say and Yakary\i lmaz \cite{SayY14C} presented a bounded--error polynomial--time 2QCFA algorithm, say $M$, for 
\[
	\powereq = \{ a b a^7 b a^{7 \cdot 8} b a^{7 \cdot 8^2} b a^{7 \cdot 8^3} b \cdots b a^{7 \cdot 8^n} \mid n \geq 0 \}.
\]
Remark that every member has $ 8^{n+1} $ $a$s for some $ n \geq 0 $. It is clear that for any member of $ \powereq $ having $ 8^{n+1} $ $a$s and for any $ I \in \mathcal{I} $, $ M $ can be modified, say $M_I$, in order to determine whether $ n $ is in $ I $ or not by using Procedure ADH with high probability. So, $ M_I $ can recognize the following language with bounded error \cite{SayY14C}
\[
	\powereqI = \{ w \in \{a,b\}^* \mid w \in \powereq \mbox{ and } \log_8( |w|_a ) \in I \}.
\]
Then, we can follow that 2QCFAs can recognize uncountably many languages with bounded error.

Procedure ADH can be implemented by a rtQCFA having a single qubit trivially, say $R_{I}$ for $ I \in \mathcal{I} $. So, if we show that $ \powereq $ is recognized by a restarting rtQCFA, say $ R $, then, we can follow that rtQCFAs can recognize uncountably many languages. Since $ R $ can execute $ R_I $ in parallel to its original algorithm, i.e. $ R $ and $ R_{ADH} $ are tensorred such that if $ R $ is in the restarting state, then all computation is restarted; otherwise, the input is accepted if and only if both $ R $ and $ R_I $ give the decision of ``accepting''. The obtained restarting rtQCFA gives its decisions with bounded error. We refer the reader to \cite{YS11B} for the technical details to obtain a restarting bounded--error rtQCFA by tensorring two bounded--error restarting rtQCFAs, where the results are given for general realtime QFA models but it can be followed for rtQCFAs in the same way since general realtime QFA models and rtQCFAs can simulate each other exactly.

\begin{theorem}
	\label{thm:rtQCFA}
	The language $ \powereq $ can be recognized by a restarting rtQCFA $ R $ with bounded error.
\end{theorem}
\begin{proof}
	The quantum part of $ R $ has 9 states ($ \{ q_1,\ldots,q_{9} \} $) but only the first three of them are used significantly. Before each unitary operation, the quantum state has always zeros after the significant first three entries
	\[
		\paran{ \alpha_1 ~~ \alpha_2 ~~ \alpha_3 ~~  0 ~~ \cdots ~~ 0 }^T.
	\]
	After we apply a unitary operator $ U $, we obtain a new quantum state. Then, we make some measurements such that if the system is in $ span\{ \ket{q_4,\ldots,\ket{q_9}} \} $, then the computation is always restarted. So, if the computation is not restarted, the new quantum state has always zeros for the last six entries.
	\[
		\paran{ \alpha'_1 ~~ \alpha'_2 ~~ \alpha'_3 ~~ 0 ~~ \cdots ~~ 0 }^T.
	\]
	Sometimes the measurement operator can also affect the first three states that will be specified later. 
	
	Each unitary operator is a $ 9 \times 9 $-dimensional unitary matrix. However, the significant parts are the top-left  $ (3 \times 3) $-dimensional matrices due to the measurement operators. So, we can trace the computation only by a 3-dimensional vector and $ (3 \times 3) $-dimensional matrices. 
	
	We describe our algorithm with integer matrices with a real coefficient $ 0 < l < 1 $:
	\[
		l A = 
		l \mymatrix{ccc}{
			  a_{11} & a_{12} & a_{13}
			\\a_{21} & a_{22} & a_{23}
			\\a_{31} & a_{32} & a_{33}
		}.
	\]
	Remark that $ lA $ is the top-left corner of a unitary matrix and all the other entries can be filled arbitrarily providing that the matrix is unitary. For our purpose, first we define our $(3 \times 3)$-dimensional matrices $A$s, which do the main tasks, and then complete the missing parts of unitary matrices by selecting a fixed $ l $ for all $A$s. We refer the reader to \cite{YS10A,YS11A} for the details of how to pick nonnegative real $ l < 1  $ and fill the missing parts of unitary matrices. 
	
	After a measurement operator, we can obtain more than one unnormalized state vector having norm less than 1. Depending on the measurement outcome, the system collapses into one of them and then the corresponding unnormalized state vector is normalized (norm-1 vector). On the other hand, since the probabilities can be calculated directly from the entries of unnormalized state vectors, we trace the computation with unnormalized state vectors. 
	
	After all these technical descriptions, we can give the details of our quantum algorithm. (QFA algorithms based on such assumptions have been presented before (e.g. see $ \cite{Yak13C} $).)
	
	The quantum part is  in state $ \ket{v_0} = \paran{1~~0~~0}^T $ at the beginning. If the input does not start with $ aba^7b $, then it is rejected deterministically. Otherwise, by reading 7 $a$s, the quantum part is set to
	\[
		\ket{\widetilde{v_7}} = l^7 \myvector{1 \\ 7 \cdot 56 \\ 0  } = \paran{ l \mymatrix{cccc}{ 1 & 0 & 0  \\ 56 & 1 & 0  \\ 0 & 0 & 0   } } ^ 7 \myvector{1 \\ 0 \\ 0 }.
	\]
	If the input is $ aba^7b $, then it is accepted deterministically. 
	
	In the remaining part, we assume that the input is of the form $ aba^7b (a^+b)^+ $. Otherwise, the input is rejected deterministically. Remark that, after each quantum step, the computation is restarted with some probability and so the computation in a single round can reach to the end-marker only with a very small (exponentially small in the input length) probability.
	
	At the beginning of each block of $ a $s, say the $i$-th block, the quantum state is
	\[
		\ket{\widetilde{v_{t_{i-1}}}} =  l^{t_{i-1}} \myvector{1 \\ 8 S_{i-1} \\ 0 },
	\]
	where 
	\begin{itemize}
		\item $ \ket{\widetilde{v_{t_{i-1}}}} $ is the non-halting \textit{unnormalized} quantum state vector,
		\item the first block ($i=1$) is the first $ a $s after $ aba^7b $,
		\item $ S_{i-1} $ is the number of $a$s in the previous block with $ S_0 = 7 $, and
		\item $ t_{i-1} = S_0 + S_1 + \cdots + S_{i-1} + i-1 $ is the number of unitary operators applied until that step, i.e. 7 unitary operators are applied on the input $ aba^7b $ (the others can be assumed as identity operator), and then for each symbol, a unitary operator is applied.
	\end{itemize}   
	Remark that the expected number of $a$s in the $i$-th block is already written as the amplitude of the $ \ket{q_2} $, i.e., for any member, the number of $a$s in the $i$-th block is 8 times of the number of $a$s in the previous block.  
	
	During reading the $i$-th block, the number of $a$s are counted and kept as the amplitude of $ \ket{q_3} $:
	\[
		l^{t_{i-1}+j} \myvector{ 1 \\ E_i \\ j } = l \mymatrix{ccc}{1 & 0 & 0 \\ 0 & 1 & 0 \\ 1 & 0 & 1} l^{t_{i-1}+j-1} \myvector{1 \\ E_i \\ j-1}.
	\]
	After reading the block, just before reading $b$, the quantum state is
	\[
		l^{t_i - 1} \myvector{1 \\ E_i \\ S_i}.
	\]
	Then, we obtain the following vector after reading $ b $ before the measurement:
	\[
		l^{t_i} \myvector{ 1 \\ E_{i+1} = 8 \cdot S_i \\ S_i - E_i } = l \mymatrix{crc}{1 & 0 & 0 \\ 0 & 0 & 8 \\ 0 & -1 & 1} l^{t_i-1} \myvector{1 \\ E_i \\ S_i}.
	\]
	After this, the measurement operator, additional to its previously described standard behaviour, also checks whether the system is in $ span\{\ket{q_1},\ket{q_2} \} $ or $ span\{ \ket{q_3} \} $. In the latter case, the input is rejected, and the computation continues, otherwise. We have two cases:
	\begin{itemize}
		\item If $ S_i \neq E_i$, then $ S_i - E_i $ is a nonzero integer, and so, the input is rejected with probability at least $ l^{2t_i} $.
		\item If $ S_i = E_i $, then the input is rejected with zero probability.
	\end{itemize} 
	That is, for any non-member, the input is rejected after a block with some nonzero probability. For each member, on the other hand, the input is never rejected until end of the computation. 
	
	At the end of the computation, the state before reading the right end-marker is 
	\[
		l^{t_k} \myvector{1 \\ 8 \cdot E_{k} \\ 0}
	\]
	if there are $ k $ blocks of $a$s. Then, we obtain the following quantum state after reading the right end-marker
	\[
		l^{t_k+1} \myvector{1 \\ 0 \\ 0} = l \mymatrix{ccc}{1 & 0 & 0 \\ 0& 0 & 0 \\ 0 & 0 & 0 } l^{t_k} \myvector{1 \\ 8 \cdot E_{k} \\ 0}
	\]
	and the input is accepted if $ \ket{q_1} $ is observed. Then the input is accepted with probability $ l^{2t_k+2} $, which is clearly at least $ l^2 $ times of any possible rejecting probability before.
	
	Now, we can analyse a single round of $ R $, the period from the initial configuration to give a decision or to restart the computation, and then calculate the overall probabilities on the given input.
	
	Any member is accepted with an exponentially small but non-zero probability ($ l^{2t_k+2} $) and it is rejected with zero probability. So, in exponential expected time, the input is accepted with probability 1. 
	
	Any non-member, on the other hand, is again accepted with a very small non-zero probability but it is also rejected with a probability sufficiently bigger than the accepting probability. So, in exponential expected time, the input is rejected with probability $ \frac{R}{A+R} $ that is at least $ \frac{R}{l^2 R + R} = \frac{1}{1+l^2} > \frac{1}{2} $, where $ A $ and $ R $ are the accepting and rejecting probabilities, respectively, in a single round (see \cite{YS10B} for the details of calculating the overall rejecting probability). Remark that $ l $ can be picked arbitrarily small and so the rejecting probability can be arbitrarily close to 1.
\end{proof}

\begin{corollary}
	Exponential expected time restarting rtQCFAs can recognize uncountably many languages with bounded error.
\end{corollary}

It is still open whether rtQCFAs can recognize a nonregular language with bounded error in polynomial time.

Some algorithms can be space sufficient only for the members. That is known as recognition with middle space \cite{Sze94}. The standard space usage is known as recognition with strong space. Until now, we focus on strong space bounds. On the other hand, we know that 2QCCAs can recognize the following nonregular unary language with bounded error in middle logarithmic space \cite{BGRY16}:
\newcommand{\upower}[1]{\mathtt{UPOWER #1}}
\[
	\upower{2} = \{ a^{2^n} \mid n \geq 0 \}.
\]
Here the base-2 is not essential and it can be replaced with any integer bigger than 2. Therefore,
2QCCAs can also recognize
\[
	\upower{8} = \{ a^{8^n} \mid n \geq 0 \}
\]
with bounded error in middle logarithmic space (by slightly modifying the algorithm for $ \upower{2} $). Moreover, for any $ I \in \mathcal{I} $, 2QCCAs recognize the following language
\[
	\upower{8(I)} = \{ a^{8^n} \mid n-1 \in I \}
\]
with bounded error in middle logarithmic space, i.e. we first determine whether the input is of the form $ a^{8^n} $ with high probability. If so, we call Procedure ADH, which does not use the counter, to determine whether $ (n-1) $ is in $ I $ or not with high probability.

\begin{theorem}
	Unary middle logarithmic--space 2QCCAs can recognize uncountably many languages with bounded error.
\end{theorem}

\section{Concluding remarks}
\label{sec:conc}

In this paper, we identify some small space (and time) bounds for bounded--error probabilistic and quantum models that can recognize uncountably many languages. We list the positive cases that we obtain below. We also present the related open cases.
\begin{itemize}
	\item Unary languages:
	\begin{itemize}
		\item Polynomial--time $O(\log \log n)$--space sweeping PTMs (open for $o(\log \log n)$-space)
		\item Linearithmic--time $O(\log n)$--space 1PTMs
		\item Middle $O(\log n)$--space 2QCCAs (open for better space bounds and/or polynomial--time; open for 2PCAs)
	\end{itemize}
	\item Binary languages:
	\begin{itemize}	
		\item $\omega(1)$--space 2PCAs (open for $O(1)$--space (or equivalently 2PFAs))
		\item Polynomial--time $O(n)$--space 2PCAs (open for polynomial--time $o(n)$--space)		
		\item Restarting rtQCFAs (open for polynomial--time)			
	\end{itemize}
\end{itemize}
The above list can be extended with some other restricted models (for example, pushdown finite automata, Turing machines with limited reversal complexity, multihead automata, constant--space interactive proof systems), which we leave as future works.

\section*{Acknowledgements}
We thank to anonymous referees for their very helpful comments.


\biblio{tcs} 

\EndOfArticle 

\end{document}